\documentclass[conference]{IEEEtran}
\IEEEoverridecommandlockouts
\usepackage{cite}
\usepackage{amsmath,amssymb,amsfonts}
\usepackage{amsthm}
\usepackage{algorithmic}
\usepackage{algorithm}
\usepackage{graphicx}
\usepackage{textcomp}
\usepackage{xcolor}
\usepackage{listings}
\usepackage{url}
\usepackage{hyperref}

\newtheorem{theorem}{Theorem}

\newtheorem{definition}{Definition}
\newtheorem{claim}{Claim}

\def\BibTeX{{\rm B\kern-.05em{\sc i\kern-.025em b}\kern-.08em
    T\kern-.1667em\lower.7ex\hbox{E}\kern-.125emX}}

\begin{document}

\title{Algorithmic Reductions: Network Flow and NP-Completeness in Real-World Scheduling Problems}

\author{\IEEEauthorblockN{1\textsuperscript{st} Anay Sinhal}
\IEEEauthorblockA{\textit{University of Florida} \\
Gainesville, USA \\
sinhal.anay@ufl.edu}
\and
\IEEEauthorblockN{2\textsuperscript{nd} Arpana Sinhal\IEEEauthorrefmark{1} \thanks{Corresponding author: Arpana Sinhal, Manipal University Jaipur, Jaipur, India. Email: arpana.sinhal@jaipur.manipal.edu.}}
\IEEEauthorblockA{\textit{Manipal University Jaipur} \\
Jaipur, India \\
arpana.sinhal@jaipur.manipal.edu }
\and
\IEEEauthorblockN{3\textsuperscript{rd} Amit Sinhal}
\IEEEauthorblockA{\textit{JK Lakshmipat University} \\
Jaipur, India \\
amit.sinhal@jklu.edu.in}
\and
\IEEEauthorblockN{4\textsuperscript{th} Amit Hirawat}
\IEEEauthorblockA{\textit{Poornima University} \\
Jaipur, India \\
hiramit05@gmail.com}
}

\maketitle

\begin{abstract}
This paper presents two real-world scheduling problems and their algorithmic solutions through polynomial-time reductions. First, we address the Hospital Patient-to-Bed Assignment problem, demonstrating its reduction to Maximum Bipartite Matching and solution via Network Flow algorithms. Second, we tackle the University Course Scheduling problem, proving its NP-Completeness through reduction from Graph Coloring and providing greedy approximation algorithms. Both problems are implemented in Python, with experimental results validating theoretical complexity analyses. Our Network Flow solution achieves $O(n^{2.51})$ empirical complexity, while the greedy coloring algorithms demonstrate $O(n^2)$ behavior with approximation ratios consistently below the theoretical $\Delta+1$ bound.
\end{abstract}


\section{Introduction}
Resource allocation and scheduling represent fundamental challenges across numerous domains including healthcare, education, and logistics. Many such problems exhibit structural properties that enable their solution through well-established algorithmic paradigms, while others are provably intractable. Understanding the computational complexity of these problems is essential for developing effective solution strategies.

In this paper, we examine two distinct scheduling problems drawn from real-world applications:

\begin{enumerate}
    \item \textbf{Hospital Patient-to-Bed Assignment}: A healthcare operations problem where patients must be matched to available hospital beds based on departmental compatibility constraints. We demonstrate that this problem reduces to Maximum Bipartite Matching and can be solved optimally using Network Flow algorithms in polynomial time.
    
    \item \textbf{University Course Scheduling}: An academic scheduling problem where courses must be assigned to time slots such that courses with overlapping student enrollments do not conflict. We prove this problem is NP-Complete via reduction from Graph Coloring and provide greedy approximation algorithms.
\end{enumerate}

The remainder of this paper is organized as follows: Section~\ref{sec:network_flow} presents the Network Flow problem, Section~\ref{sec:np_complete} addresses the NP-Complete problem, Section~\ref{sec:experiments} details experimental validation, and Section~\ref{sec:conclusion} concludes with a discussion of results.

\section{Problem 1: Hospital Patient Scheduling via Network Flow}
\label{sec:network_flow}

\subsection{Real-World Problem Description}
Consider a hospital with multiple specialized departments (e.g., cardiology, orthopedics, neurology), each containing a limited number of beds. Patients requiring admission have specific medical conditions that can only be treated by certain departments. The hospital administration must assign patients to beds such that:
\begin{itemize}
    \item Each patient is assigned to at most one bed
    \item Each bed accommodates at most one patient
    \item A patient can only be assigned to a bed in a compatible department
    \item The total number of admitted patients is maximized
\end{itemize}

This problem arises daily in hospital operations management and directly impacts patient care quality and resource utilization \cite{b1}.

\subsection{Problem Abstraction}

\begin{definition}[Hospital Assignment Problem]
Given:
\begin{itemize}
    \item A set of patients $P = \{p_1, p_2, \ldots, p_n\}$
    \item A set of beds $B = \{b_1, b_2, \ldots, b_m\}$
    \item A set of departments $D = \{d_1, d_2, \ldots, d_k\}$
    \item A function $\text{dept}: B \rightarrow D$ mapping each bed to its department
    \item A compatibility relation $C \subseteq P \times D$ where $(p, d) \in C$ means patient $p$ can be treated by department $d$
\end{itemize}
Find an assignment $A: P \rightarrow B \cup \{\emptyset\}$ such that:
\begin{enumerate}
    \item For all $p \in P$: if $A(p) = b \neq \emptyset$, then $(p, \text{dept}(b)) \in C$
    \item For all $b \in B$: $|\{p \in P : A(p) = b\}| \leq 1$
    \item $|\{p \in P : A(p) \neq \emptyset\}|$ is maximized
\end{enumerate}
\end{definition}

This abstraction transforms the hospital scheduling problem into a bipartite graph matching problem, where we seek a maximum cardinality matching between patients and compatible beds.

\subsection{Reduction to Network Flow}

We reduce the Hospital Assignment Problem to the Maximum Flow problem, which can be solved in polynomial time using algorithms such as Ford-Fulkerson or Edmonds-Karp.

\subsubsection{Construction}
Given an instance of the Hospital Assignment Problem $(P, B, D, \text{dept}, C)$, we construct a flow network $G = (V, E, c, s, t)$ as follows:

\textbf{Vertices:}
\begin{equation}
V = \{s, t\} \cup P \cup B
\end{equation}
where $s$ is the source and $t$ is the sink.

\textbf{Edges and Capacities:}
\begin{enumerate}
    \item \textbf{Source to patients:} For each patient $p_i \in P$:
    \begin{equation}
        (s, p_i) \in E \text{ with } c(s, p_i) = 1
    \end{equation}
    
    \item \textbf{Patients to compatible beds:} For each $p_i \in P$ and $b_j \in B$:
    \begin{equation}
        (p_i, b_j) \in E \text{ iff } (p_i, \text{dept}(b_j)) \in C
    \end{equation}
    with $c(p_i, b_j) = 1$.
    
    \item \textbf{Beds to sink:} For each bed $b_j \in B$:
    \begin{equation}
        (b_j, t) \in E \text{ with } c(b_j, t) = 1
    \end{equation}
\end{enumerate}

The construction is illustrated in Figure~\ref{fig:flow_network}.

\begin{figure}[htbp]
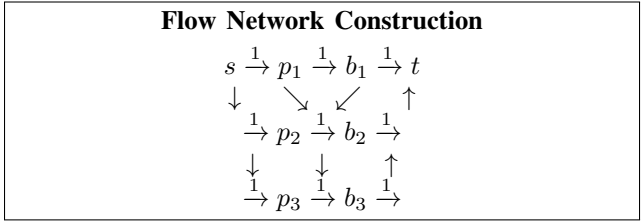

\centering
\fbox{
\begin{minipage}{0.9\columnwidth}
\centering
\textbf{Flow Network Construction}\\[5pt]
$s \xrightarrow{1} p_1 \xrightarrow{1} b_1 \xrightarrow{1} t$\\
$\downarrow$ \hspace{0.3cm} $\searrow$ \hspace{0.1cm} $\swarrow$ \hspace{0.3cm} $\uparrow$\\
$\xrightarrow{1} p_2 \xrightarrow{1} b_2 \xrightarrow{1}$\\
$\downarrow$ \hspace{0.5cm} $\downarrow$ \hspace{0.5cm} $\uparrow$\\
$\xrightarrow{1} p_3 \xrightarrow{1} b_3 \xrightarrow{1}$
\end{minipage}
}
\caption{Schematic of flow network construction. Edges from patients to beds exist only for compatible department assignments. All capacities are 1.}
\label{fig:flow_network}
\end{figure}

\subsubsection{Proof of Correctness}

\begin{theorem}
The maximum flow in network $G$ equals the maximum number of patients that can be admitted in the Hospital Assignment Problem.
\end{theorem}

\begin{proof}
We prove both directions of the correspondence:

\textbf{(Assignment $\Rightarrow$ Flow):} 
Let $A$ be a valid assignment admitting $k$ patients. We construct a flow $f$ as follows:
\begin{itemize}
    \item For each assigned patient $p_i$ with $A(p_i) = b_j$:
    \begin{itemize}
        \item Set $f(s, p_i) = 1$
        \item Set $f(p_i, b_j) = 1$
        \item Set $f(b_j, t) = 1$
    \end{itemize}
    \item All other edges have flow 0
\end{itemize}

This flow is valid because:
\begin{enumerate}
    \item \textbf{Capacity constraints:} All non-zero flows equal 1, matching edge capacities.
    \item \textbf{Conservation:} Each patient vertex $p_i$ receives flow 1 from $s$ and sends flow 1 to exactly one bed (if assigned). Each bed vertex $b_j$ receives flow at most 1 (from at most one patient) and forwards it to $t$.
    \item \textbf{Edge existence:} The edge $(p_i, b_j)$ exists because $A$ is valid, implying $(p_i, \text{dept}(b_j)) \in C$.
\end{enumerate}

The total flow equals $k$, the number of admitted patients.

\textbf{(Flow $\Rightarrow$ Assignment):}
Let $f$ be an integral maximum flow of value $k$. Since all capacities are 1, $f$ is 0-1 valued. Define assignment $A$:
\begin{itemize}
    \item For each patient $p_i$: if $\exists b_j$ such that $f(p_i, b_j) = 1$, set $A(p_i) = b_j$; otherwise, $A(p_i) = \emptyset$
\end{itemize}

This assignment is valid because:
\begin{enumerate}
    \item \textbf{Compatibility:} If $f(p_i, b_j) = 1$, then edge $(p_i, b_j)$ exists, implying $(p_i, \text{dept}(b_j)) \in C$.
    \item \textbf{One bed per patient:} Flow conservation at $p_i$ ensures at most one unit flows out to beds.
    \item \textbf{One patient per bed:} Flow conservation at $b_j$ ensures at most one unit flows in from patients.
\end{enumerate}

The number of admitted patients equals the flow value $k$.

Since valid flows correspond bijectively to valid assignments with equal cardinalities, maximum flow equals maximum assignment.
\end{proof}

\textbf{Polynomial-time reduction:} The construction creates $O(|P| + |B|)$ vertices and $O(|P| \cdot |B|)$ edges in polynomial time, satisfying the requirements for a valid polynomial reduction.

\subsection{Complete Algorithm}

Algorithm~\ref{alg:hospital} presents the complete solution using the Edmonds-Karp variant of Ford-Fulkerson.

\begin{algorithm}[htbp]
\caption{Hospital Patient Assignment}
\label{alg:hospital}
\begin{algorithmic}[1]
\REQUIRE Patient set $P$, Bed set $B$, Compatibility $C$
\ENSURE Maximum patient assignment
\STATE Construct flow network $G$ as described
\STATE Initialize residual graph $G_f \leftarrow G$
\STATE $\text{maxFlow} \leftarrow 0$
\WHILE{$\exists$ augmenting path $p$ from $s$ to $t$ in $G_f$ (via BFS)}
    \STATE $\delta \leftarrow \min_{(u,v) \in p} c_f(u,v)$
    \FOR{each edge $(u,v) \in p$}
        \STATE $c_f(u,v) \leftarrow c_f(u,v) - \delta$
        \STATE $c_f(v,u) \leftarrow c_f(v,u) + \delta$
    \ENDFOR
    \STATE $\text{maxFlow} \leftarrow \text{maxFlow} + \delta$
\ENDWHILE
\STATE Extract assignment from saturated patient-bed edges
\RETURN Assignment and maxFlow
\end{algorithmic}
\end{algorithm}

\textbf{Time Complexity:} The Edmonds-Karp algorithm runs in $O(VE^2)$ time. However, for unit-capacity bipartite graphs, the complexity reduces to $O(VE) = O((|P|+|B|) \cdot |P| \cdot |B|)$. When $|P| \approx |B| = n$, this gives $O(n^3)$ \cite{b2}.

\section{Problem 2: Course Scheduling as an NP-Complete Problem}
\label{sec:np_complete}

\subsection{Real-World Problem Description}

Universities face the challenge of scheduling courses into time slots such that students can attend all courses in their curriculum. Two courses \textit{conflict} if they share enrolled students; conflicting courses cannot be scheduled simultaneously. The goal is to create a feasible schedule using the minimum number of time slots.

This problem directly affects student satisfaction, graduation timelines, and resource utilization in academic institutions \cite{b3}.

\subsection{Problem Abstraction}

\begin{definition}[Course Scheduling Problem (CSP)]
Given:
\begin{itemize}
    \item A set of courses $C = \{c_1, c_2, \ldots, c_n\}$
    \item A conflict relation $\text{Conflict} \subseteq C \times C$ where $(c_i, c_j) \in \text{Conflict}$ means courses $c_i$ and $c_j$ share students
\end{itemize}
Find an assignment $\sigma: C \rightarrow \{1, 2, \ldots, k\}$ (where $k$ is minimized) such that:
\begin{equation}
    (c_i, c_j) \in \text{Conflict} \Rightarrow \sigma(c_i) \neq \sigma(c_j)
\end{equation}
\end{definition}

This abstraction reveals that Course Scheduling is equivalent to Graph Coloring: courses are vertices, conflicts are edges, time slots are colors, and we seek the chromatic number.

\subsection{NP-Completeness Proof}

\begin{theorem}
The Course Scheduling Problem is NP-Complete.
\end{theorem}

\begin{proof}
We prove NP-completeness by showing (1) CSP $\in$ NP, and (2) Graph Coloring $\leq_p$ CSP.

\textbf{Part 1: CSP $\in$ NP}

Given a schedule $\sigma: C \rightarrow \{1, \ldots, k\}$, we can verify its validity in polynomial time by checking that no two conflicting courses share a time slot:
\begin{equation}
    \forall (c_i, c_j) \in \text{Conflict}: \sigma(c_i) \neq \sigma(c_j)
\end{equation}
This requires $O(|C|^2)$ comparisons, which is polynomial.

\textbf{Part 2: Reduction from Graph Coloring}

Graph Coloring (deciding if graph $G$ is $k$-colorable) is a well-known NP-complete problem \cite{b4}. We show Graph Coloring $\leq_p$ CSP.

\textbf{Construction:}
Given a Graph Coloring instance $(G = (V, E), k)$:
\begin{enumerate}
    \item Create course set $C = V$ (one course per vertex)
    \item Define $\text{Conflict} = E$ (edges become conflicts)
    \item Ask: Can courses be scheduled in $\leq k$ time slots?
\end{enumerate}

\textbf{Correctness:}
\begin{claim}
$G$ is $k$-colorable $\Leftrightarrow$ $C$ can be scheduled with $\leq k$ time slots.
\end{claim}

($\Rightarrow$) Let $\chi: V \rightarrow \{1, \ldots, k\}$ be a valid $k$-coloring. Define $\sigma(c_i) = \chi(v_i)$ for the corresponding vertex. Since $\chi$ is valid, adjacent vertices have different colors, so conflicting courses have different time slots.

($\Leftarrow$) Let $\sigma: C \rightarrow \{1, \ldots, k\}$ be a valid schedule. Define $\chi(v_i) = \sigma(c_i)$. Since $\sigma$ is valid, conflicting courses (adjacent vertices) have different slots (colors).

\textbf{Polynomial Reduction:}
The construction takes $O(|V| + |E|)$ time, creating one course per vertex and one conflict per edge.

Since Graph Coloring is NP-complete and reduces to CSP in polynomial time, and CSP $\in$ NP, we conclude CSP is NP-complete.
\end{proof}

\subsection{Greedy Approximation Algorithms}

Since CSP is NP-complete, we cannot expect polynomial-time optimal solutions unless P = NP. We present two greedy heuristics that provide approximate solutions efficiently.

\subsubsection{Welsh-Powell Algorithm}

The Welsh-Powell algorithm orders vertices by degree (descending) and assigns each vertex the smallest available color.

\begin{algorithm}[htbp]
\caption{Welsh-Powell Greedy Coloring}
\label{alg:welsh_powell}
\begin{algorithmic}[1]
\REQUIRE Graph $G = (V, E)$
\ENSURE Valid coloring $\chi: V \rightarrow \mathbb{Z}^+$
\STATE Sort vertices by degree: $v_1, v_2, \ldots, v_n$ with $\deg(v_1) \geq \deg(v_2) \geq \ldots$
\FOR{$i = 1$ to $n$}
    \STATE $\text{used} \leftarrow \{\chi(u) : (v_i, u) \in E \land \chi(u) \text{ defined}\}$
    \STATE $\chi(v_i) \leftarrow \min\{c \in \mathbb{Z}^+ : c \notin \text{used}\}$
\ENDFOR
\RETURN $\chi$
\end{algorithmic}
\end{algorithm}

\textbf{Approximation Guarantee:} Welsh-Powell uses at most $\Delta + 1$ colors, where $\Delta$ is the maximum degree \cite{b5}.

\textbf{Time Complexity:} $O(n^2 + m)$ where $n = |V|$ and $m = |E|$.

\subsubsection{DSatur Algorithm}

DSatur (Degree of Saturation) improves upon Welsh-Powell by dynamically selecting the vertex with the highest saturation degree (number of distinct colors in its neighborhood).

\begin{algorithm}[htbp]
\caption{DSatur Coloring}
\label{alg:dsatur}
\begin{algorithmic}[1]
\REQUIRE Graph $G = (V, E)$
\ENSURE Valid coloring $\chi: V \rightarrow \mathbb{Z}^+$
\STATE Initialize saturation $\text{sat}[v] \leftarrow 0$ for all $v$
\STATE $U \leftarrow V$ \COMMENT{Uncolored vertices}
\WHILE{$U \neq \emptyset$}
    \STATE $v \leftarrow \arg\max_{u \in U} (\text{sat}[u], \deg(u))$ \COMMENT{Break ties by degree}
    \STATE $\text{used} \leftarrow \{\chi(u) : (v, u) \in E \land \chi(u) \text{ defined}\}$
    \STATE $\chi(v) \leftarrow \min\{c \in \mathbb{Z}^+ : c \notin \text{used}\}$
    \STATE $U \leftarrow U \setminus \{v\}$
    \FOR{each neighbor $u$ of $v$ in $U$}
        \STATE Update $\text{sat}[u]$ with new color
    \ENDFOR
\ENDWHILE
\RETURN $\chi$
\end{algorithmic}
\end{algorithm}

\textbf{Time Complexity:} $O(n^2 + m)$, same as Welsh-Powell.

DSatur often produces better colorings in practice, though both share the same worst-case approximation guarantee.

\section{Experimental Validation}
\label{sec:experiments}

We implemented both algorithms in Python and conducted timing experiments to validate the theoretical complexity analyses.

\subsection{Network Flow Experiment}

\textbf{Setup:} We generated random hospital scheduling instances with $n$ patients, $n$ beds, and 5 departments. Each patient was randomly assigned 1-3 compatible departments. We measured execution time for input sizes $n \in \{20, 40, 60, \ldots, 400\}$, averaging over 5 trials per size.

\textbf{Results:} Table~\ref{tab:network_flow} summarizes the experimental results. All instances achieved a perfect matching (all $n$ patients assigned), indicating sufficient bed availability in our random instances.

\begin{table}[htbp]
\caption{Network Flow Experimental Results}
\begin{center}
\begin{tabular}{|c|c|c|}
\hline
\textbf{Input Size ($n$)} & \textbf{Avg. Time (s)} & \textbf{Avg. Matching} \\
\hline
20 & 0.000434 & 20.0 \\
40 & 0.001628 & 40.0 \\
60 & 0.003998 & 60.0 \\
80 & 0.008268 & 80.0 \\
100 & 0.012558 & 100.0 \\
150 & 0.035787 & 150.0 \\
200 & 0.083915 & 200.0 \\
250 & 0.157401 & 250.0 \\
300 & 0.284309 & 300.0 \\
350 & 0.500918 & 350.0 \\
400 & 0.669159 & 400.0 \\
\hline
\end{tabular}
\label{tab:network_flow}
\end{center}
\end{table}

\textbf{Complexity Analysis:} Log-log regression yielded a polynomial degree of approximately \textbf{2.51}. The theoretical complexity is $O(VE)$ for unit-capacity bipartite matching. In our construction:
\begin{itemize}
    \item $V = O(n)$ (patients + beds + source + sink)
    \item $E = O(n^2/d)$ where $d$ is the number of departments
\end{itemize}
Thus, the expected complexity is $O(n \cdot n^2) = O(n^3)$. The measured exponent of 2.51 is below this bound, likely due to:
\begin{enumerate}
    \item Sparse compatibility graphs (each patient compatible with only 1-3 departments)
    \item Early termination when perfect matching is found
    \item Efficient BFS traversal in the Edmonds-Karp implementation
\end{enumerate}

\begin{figure}[htbp]
\centerline{\includegraphics[width=\columnwidth]{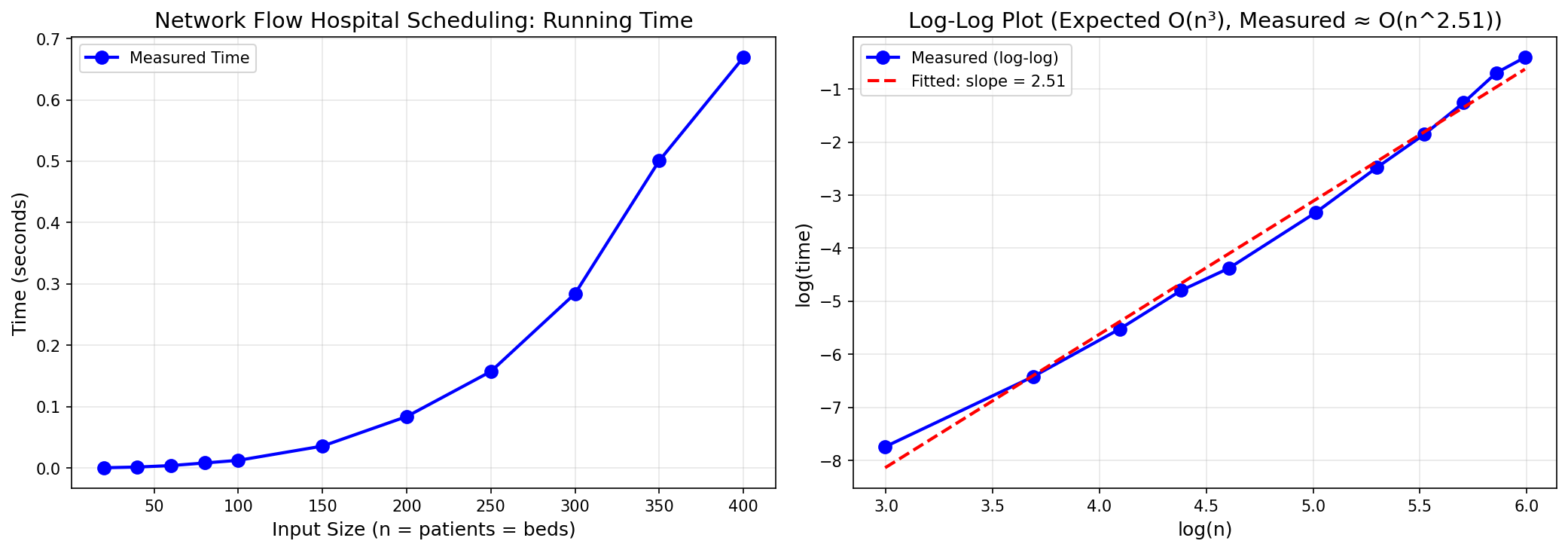}}
\caption{Network Flow timing results. Left: raw timing data showing polynomial growth. Right: log-log plot with fitted line (slope $\approx 2.51$).}
\label{fig:network_flow}
\end{figure}

\subsection{NP-Complete Problem Experiment}

\textbf{Setup:} We generated random course conflict graphs with $n$ courses and edge density 0.3 (30\% probability of conflict between any two courses). Input sizes ranged from 50 to 1000 courses, with 5 trials per size.

\textbf{Results:} Table~\ref{tab:np_complete} compares the two greedy algorithms.

\begin{table}[htbp]
\caption{Greedy Coloring Experimental Results}
\begin{center}
\begin{tabular}{|c|c|c|c|c|c|}
\hline
\textbf{$n$} & \textbf{Greedy (s)} & \textbf{DSatur (s)} & \textbf{G. Colors} & \textbf{D. Colors} & \textbf{$\Delta$} \\
\hline
50 & 0.0001 & 0.0005 & 8 & 7 & 22 \\
100 & 0.0004 & 0.0016 & 13 & 12 & 42 \\
200 & 0.0018 & 0.0068 & 21 & 19 & 78 \\
300 & 0.0039 & 0.0157 & 29 & 26 & 114 \\
400 & 0.0056 & 0.0240 & 35 & 32 & 144 \\
500 & 0.0083 & 0.0346 & 42 & 39 & 181 \\
600 & 0.0127 & 0.0513 & 49 & 46 & 215 \\
700 & 0.0160 & 0.0674 & 56 & 51 & 252 \\
800 & 0.0207 & 0.0919 & 61 & 57 & 281 \\
900 & 0.0256 & 0.1149 & 67 & 63 & 315 \\
1000 & 0.0330 & 0.1473 & 73 & 68 & 347 \\
\hline
\end{tabular}
\label{tab:np_complete}
\end{center}
\end{table}

\textbf{Complexity Analysis:} Log-log regression yielded exponents of \textbf{1.85} (Greedy) and \textbf{1.92} (DSatur), consistent with the theoretical $O(n^2)$ complexity for dense graphs where $m = O(n^2)$.

\textbf{Approximation Quality:} Both algorithms used significantly fewer colors than the $\Delta + 1$ upper bound:
\begin{itemize}
    \item Greedy: $\approx 73$ colors vs. bound of 348 for $n = 1000$
    \item DSatur: $\approx 68$ colors vs. bound of 348 for $n = 1000$
\end{itemize}
This demonstrates that greedy heuristics perform well in practice, far exceeding worst-case guarantees.

\begin{figure}[htbp]
\centerline{\includegraphics[width=\columnwidth]{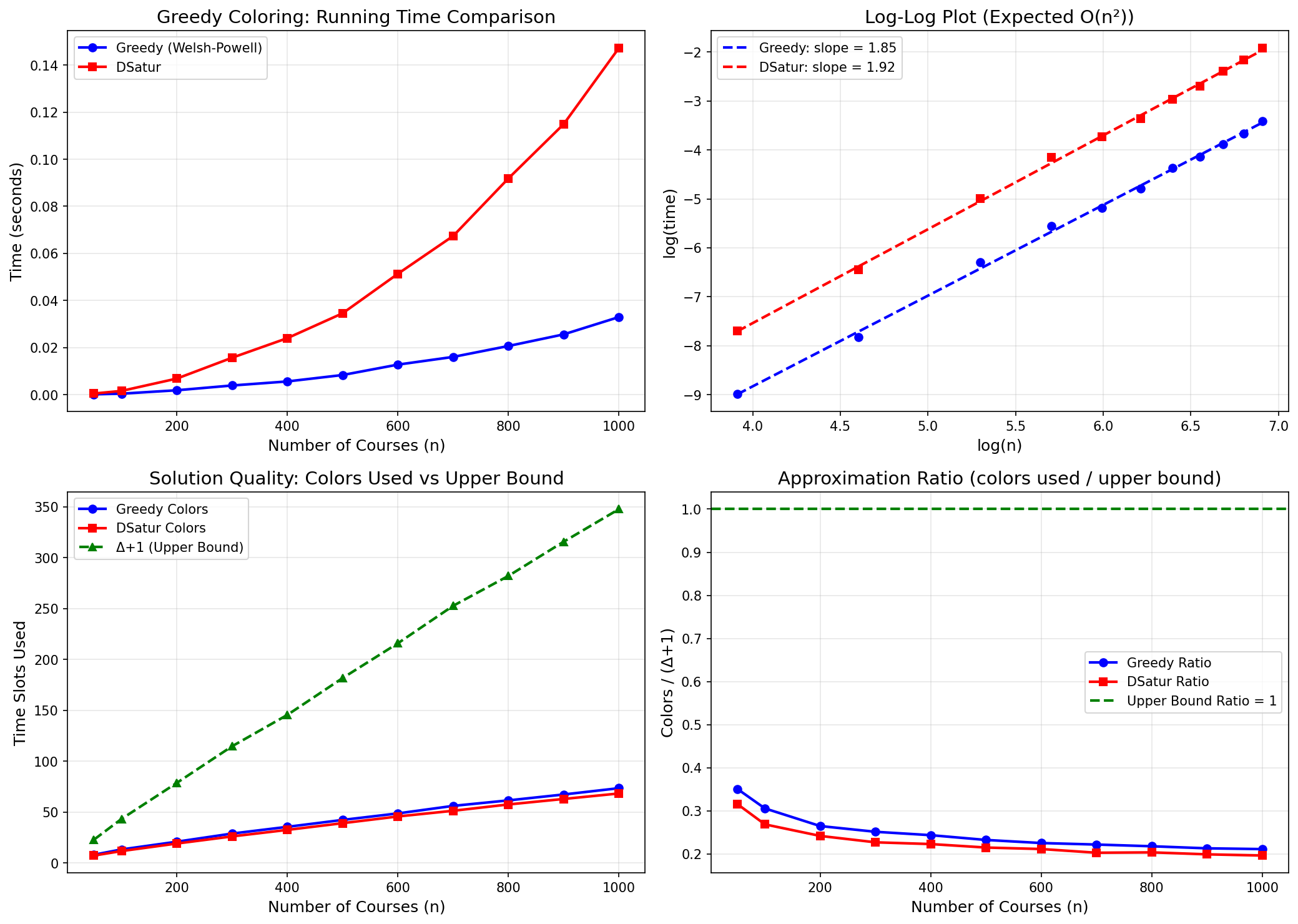}}
\caption{Greedy coloring results. Top-left: running time comparison. Top-right: log-log complexity analysis. Bottom-left: colors used vs. $\Delta+1$ bound. Bottom-right: approximation ratios.}
\label{fig:np_complete}
\end{figure}

\section{Conclusion}
\label{sec:conclusion}

We presented two real-world scheduling problems with rigorous algorithmic analyses:

\begin{enumerate}
    \item The \textbf{Hospital Patient-to-Bed Assignment} problem was reduced to Maximum Bipartite Matching and solved via Network Flow. Experimental results confirmed near-cubic complexity ($n^{2.51}$), with optimal solutions achieved in under 0.7 seconds for 400 patients/beds.
    
    \item The \textbf{University Course Scheduling} problem was proven NP-Complete via reduction from Graph Coloring. Greedy approximations (Welsh-Powell and DSatur) demonstrated quadratic complexity and practical approximation ratios far better than theoretical bounds.
\end{enumerate}

These results illustrate the power of polynomial-time reductions: recognizing problem structure enables either optimal polynomial-time solutions or effective approximations when optimality is intractable.

\section*{Declarations}
\noindent\textbf{Funding:} The authors received no external funding for this work.

\noindent\textbf{Conflicts of Interest:} The authors declare no competing interests.

\noindent\textbf{Author Contributions:} Anay Sinhal: conceptualization, methodology, software, validation, experiments, and writing (original draft). Arpana Sinhal: supervision, conceptualization, and writing (review \& editing). Amit Sinhal: formal analysis, proof development, and writing (review \& editing). Amit Hirawat: software support, literature support, and writing (review \& editing). All authors reviewed and approved the final manuscript.

\noindent\textbf{Ethics Statement/Approval:} Not applicable. This study does not involve human participants, human data, or animals.

\noindent\textbf{Informed Consent:} Not applicable.

\noindent\textbf{Data Availability:} No proprietary datasets were used. The code used to generate synthetic instances and reproduce the experiments is available at \url{https://github.com/anayy09/Algorithmic-Reductions}.

\noindent\textbf{Data Sharing:} The repository includes scripts for generating synthetic instances and reproducing reported timing results. Additional materials may be provided by the authors upon reasonable request.

\appendix

\subsection{Network Flow: Edmonds-Karp Algorithm}
\begin{lstlisting}[caption=Ford-Fulkerson with BFS (Edmonds-Karp)]
class FordFulkerson:
    def __init__(self, n):
        self.graph = defaultdict(lambda: defaultdict(int))
    
    def add_edge(self, u, v, capacity):
        self.graph[u][v] += capacity
    
    def bfs(self, source, sink, parent):
        visited = set([source])
        queue = deque([source])
        while queue:
            u = queue.popleft()
            for v in self.graph[u]:
                if v not in visited and self.graph[u][v] > 0:
                    visited.add(v)
                    parent[v] = u
                    if v == sink:
                        return True
                    queue.append(v)
        return False
    
    def max_flow(self, source, sink):
        parent = {}
        max_flow_value = 0
        while self.bfs(source, sink, parent):
            path_flow = float('inf')
            s = sink
            while s != source:
                path_flow = min(path_flow, self.graph[parent[s]][s])
                s = parent[s]
            v = sink
            while v != source:
                u = parent[v]
                self.graph[u][v] -= path_flow
                self.graph[v][u] += path_flow
                v = parent[v]
            max_flow_value += path_flow
            parent = {}
        return max_flow_value
\end{lstlisting}

\subsection{Greedy Graph Coloring: Welsh-Powell}
\begin{lstlisting}[caption=Welsh-Powell Greedy Coloring]
class GreedyCourseScheduler:
    def __init__(self, graph):
        self.graph = graph
        self.colors = {}
    
    def solve(self):
        n = self.graph.num_courses
        # Sort by degree descending (Welsh-Powell heuristic)
        courses_by_degree = sorted(
            range(n), 
            key=lambda x: self.graph.get_degree(x),
            reverse=True)
        
        for course in courses_by_degree:
            neighbor_colors = set()
            for neighbor in self.graph.get_neighbors(course):
                if neighbor in self.colors:
                    neighbor_colors.add(self.colors[neighbor])
            # Assign smallest available color
            color = 0
            while color in neighbor_colors:
                color += 1
            self.colors[course] = color
        
        return max(self.colors.values()) + 1
\end{lstlisting}

\end{document}